\pdfoutput=1
\documentclass[a4paper,USenglish]{lipics-v2019}

\usepackage[noend]{algpseudocode}
\usepackage{algorithm}

\usepackage{cutwin}
\usepackage{adjustbox}
\usepackage{wrapfig}
\usepackage{xspace}
\usepackage{dsfont}

\usepackage{algorithm}
\usepackage{amsfonts}
\usepackage{amssymb}
\usepackage{hyperref}
\usepackage{amsmath}
\usepackage{cite}

\usepackage{mathrsfs}
\usepackage{graphicx}
\usepackage{enumerate}
\usepackage{enumitem} 
\usepackage{tcolorbox}

\DeclareMathOperator{\argmax}{argmax}
\DeclareMathOperator{\argmin}{argmin}
\DeclareMathOperator{\poly}{poly}
\DeclareMathOperator{\polylog}{polylog}

\newenvironment{claim-repeat}[1]{\begin{trivlist}
		\item[\hspace{\labelsep}{\indent\sc{Claim} \ref{#1} }]\em }%
	{\end{trivlist}}
\newenvironment{lemma-repeat}[1]{\begin{trivlist}
		\item[\hspace{\labelsep}{\indent\sc{Lemma} \ref{#1} }]\em }%
	{\end{trivlist}}
\newenvironment{theorem-repeat}[1]{\begin{trivlist}
		\item[\hspace{\labelsep}{\indent\sc{Theorem} \ref{#1} }]\em }%
	{\end{trivlist}}

\newcommand{\size}[1]{\ensuremath{\left|#1\right|}}

\newcommand{\set}[1]{\left\{ #1 \right\}}
\newcommand{\Gcal}{\mathcal{G}}
\newcommand{\Hcal}{\mathcal{H}}
\newcommand{\ceil}[1]{\left\lceil #1 \right\rceil}
\newcommand{\floor}[1]{\left\lfloor #1 \right\rfloor}
\newcommand{\RR}{\mathds{R}}
\newcommand{\NN}{\mathds{N}}

\newcommand{\gtmp}{G_{\mathrm{adv}}}
\newcommand{\gnew}{G_{\mathrm{new}}}

\newcommand{\gallowed}{\Gcal_{\mathrm{allowed}}}
\newcommand{\wold}{w_\mathrm{old}}
\newcommand{\wnew}{w_\mathrm{new}}

\DeclareMathOperator{\roundp}{\mathrm{round}}

\newcommand{\slow}{\ensuremath{\mathrm{slow}}}
\newcommand{\new}{\mathrm{new}}

\newcommand{\alggapreduce}{{\textnormal{\sf{GapReduce}}}}

\usepackage{listings}

\AtBeginDocument{%
  \providecommand\BibTeX{{%
    \normalfont B\kern-0.5em{\scshape i\kern-0.25em b}\kern-0.8em\TeX}}}

\title{On the Complexity of Load Balancing in Dynamic Networks} 

\author
{Seth Gilbert}
{Department of Computer Science, National University of Singapore, Sungapure}
{}
{}
{Supported, in part, by the MOE Tier 2 project “Beyond Worst-Case Analysis: A Tale of Distributed Algorithms” [MOE2018-T2-1-160]}

\author
{Uri Meir}
{School of Computer Science, Tel Aviv University, Israel}
{}
{}
{} 

\author
{Ami Paz}
{Faculty of Computer Science, Universit\"at Wien, Austria}
{}
{}
{We acknowledge the Austrian Science Fund (FWF) and netIDEE SCIENCE project P~33775-N}

\author
{Gregory Schwartzman}
{School of Information Science, Japan Advanced Institute of Science and Technology, Japan}
{}
{}
{Supported by JSPS KAKENHI grants 19K20216 and 21K17703} 

\authorrunning{S. Gilbert, U. Meir, A. Paz, and G. Schwartzman}

\Copyright{Seth Gilbert, Uri Meir, Ami Paz, and Gregory Schwartzman} 

\acknowledgements{We thank an anonymous SPAA'21 reviewer for referring us to~\cite{BerenbrinkFKK19}.}

\nolinenumbers
\hideLIPIcs  



\begin{document}
	\maketitle
\begin{abstract}
    In the load balancing problem, each node in a network is assigned a load, and the goal is to equally distribute the loads among the nodes, by preforming local load exchanges.
    While load balancing was extensively studied in static networks, only recently a load balancing algorithm for dynamic networks with a \emph{bounded convergence time} was presented.
    In this paper, we further study the time complexity of load balancing in the context of dynamic networks. 

    First, we show that randomness is not necessary, and present a \emph{deterministic} algorithm which slightly improves the running time of the previous algorithm, at the price of not being matching-based.
    Then, we consider \emph{integral} loads, i.e., loads that cannot be split indefinitely, and prove that no matching-based algorithm can have a bounded convergence time for this case.
    
    To circumvent both this impossibility result, and a known one for the non-integral case, we apply the method of \emph{smoothed analysis}, where random perturbations are made over the worst-case choices of network topologies. We show both impossibility results do not hold under this kind of analysis, suggesting that load-balancing in real world systems might be faster than the lower bounds suggest.
\end{abstract}

\section{Introduction}

Load balancing is a fundamental problem in distributed systems: a collection of computational entities has a large amount of work that must be processed; the entities can allocate the work amongst themselves to share the load evenly (e.g., to more efficiently utilize the computational power of the network, or to minimize the time necessary to accomplish it all), using a communication network to transfer load. Load balancing arises in various settings such as multiprocessor systems~\cite{Cybenko89}, virtual machine placement in data centers~\cite{Hyser2007} and mobile peer to peer networks~\cite{BaranasuriyaGNR14, NandanDPGS05, CornejoGN12}.

When balancing load, a key question is how information regarding the load is exchanged and how the load itself is transferred. How do devices coordinate to decide where to move the load?  Which devices can exchange load at a given time?  The challenge addressed in this paper is accomplishing distributed load balancing, where devices can only communicate with their neighbors, in a \emph{dynamic network}, where the underlying communication network is constantly changing.
We consider two fundamental variants of the problem:
\emph{integral}, where the loads cannot be split indefinitely;
and \emph{continuous}, where loads can be split into arbitrarily small units.

\medskip
\noindent\textbf{Locality.} Our work focuses on \emph{local} load balancing: at any given time, each device has a set of neighbors that it can communicate with, i.e., there is a fixed communication graph.  A node can exchange information and work only with its immediate neighbors, i.e., the process is entirely decentralized.  This approach is very natural in distributed systems where no central entity exists. Local load balancing was first formalized and analyzed by Cybenko \cite{Cybenko89} and Boillat \cite{Boillat90}. It has since been studied in both static \cite{Cybenko89, Boillat90, HosseiniLMMV90, XuL93, XuL94} and dynamic graph topologies \cite{BahiCV03, BahiCV05, BahiCV07, ElsasserMS04, SiderC09}.

An interesting observation is that metadata related to load may be much smaller than the information required to handle the load itself (e.g., the information describing the tasks). 
Devices can rapidly exchange information with all their neighbors about the amount of load in the system, while transferring the load is more expensive. We therefore assume that time is divided into rounds, where devices can exchange \emph{information} with all their neighbors in a round (as in classical distributed-computing models such as the \textsc{congest} model), and can exchange \emph{load} with only $O(1)$ neighbors in a round.
This distinction between coordinating information and transferring load was described in~\cite{DinitzFGN17} in the context of radio networks\footnote{The idea that it is easy to exchange a small amount of information outside of the main communication channel was also discussed in~\cite{GhaffariN16}, though in a different context.}, where each node is allowed to exchange load with a single neighbor in each round (\emph{matching-based}). 
In the integral case, we consider \emph{matching-based} algorithms;
in the continuous case, we allow a node to exchange load with \emph{two neighbors} at each round.

%
%

\medskip
\noindent\textbf{Dynamic networks.} In this work, we focus on balancing load in \emph{dynamic} networks, i.e., networks where the communication graph changes in every round. This fact makes the problem significantly more challenging than in static networks. Consider, for example, a communication graph consisting merely of a line (a simple path). In a static graph, it is easy to distribute the load evenly among all the nodes in $O(n)$ time.  For example, if all the load starts at the left end of the path, then it is sufficient to simply pass the proper amount of load to the right. 
On the other hand, in a dynamic graph, the problem is more challenging.  For example, at every step the adversary may ``re-sort'' the graph so that the nodes are ordered from the highest amount of load to the lowest. 
Now any attempt to move load down the line is immediately undone by the reorganization of the nodes!
In fact, in~\cite{DinitzFGN17}, it was shown (using a slightly more complicated example) that any randomized matching-based load balancing algorithm would require $\Omega(n^2)$ rounds to achieve a good level of balance.

\medskip
\noindent\textbf{Balance.} The goal of our paper is to achieve uniform balance across all nodes in the network---while at the static setting it makes sense to consider \emph{local} balance, i.e., the requirement that any two neighbors have a similar load, in a dynamic network the set of neighbors changes and studying \emph{global} balance is inevitable.
In some cases, perfectly uniform balance may be stronger than necessary, and so we only require approximate balance. We say that an algorithm achieves $\tau$-convergence if the following holds: for every two nodes $u,v$, if $w(u)$ is the amount of load at $u$ and $w(v)$ is the amount of load at $v$, then $|w(u)-w(v)| \leq \tau$.

\medskip
\noindent\textbf{Smoothed analysis.} Smoothed analysis is a framework suggested for mediating gaps between worst-case and average-case complexities. The underlying motivation for this form of complexity analysis is that indeed the input for the algorithm (in our case a series of network topologies) can be chosen by a very strong (adaptive) adversary, but it is not \emph{all powerful}. 
That is, at every time step the adversary changes the network topology, but some noise is added to this topology in the form of a few random edges being added or deleted. 
The algorithm then receives this new, slightly perturbed topology as input rather than the original (adversarial) one.

\subsection{Results}
Dinitz et al.~\cite{DinitzFGN17} were the first to show that continuous load-balancing has a bounded convergence time in dynamic networks; here, we focus on the \emph{time complexity} of the load balancing problem in dynamic networks. That is, we aim to achieve a deeper theoretical understanding of this problem, and improved convergence time guarantees. 
To this end, we consider the following three fundamental questions: 
\begin{enumerate}
    \item \textbf{What is the role of \emph{randomness} in achieving a fast convergence time for the problem?} We show that when it is possible to exchange load with $O(1)$ neighbors, 
    there exists a \emph{deterministic} algorithm for the problem that is \emph{faster} than the randomized algorithm of~\cite{DinitzFGN17}.
    
    \item \textbf{Assuming \emph{integral} loads, can a bounded convergence time still be achieved?} We show that any matching-based load balancing algorithm \emph{cannot} achieve a bounded convergence time.
    
    \item \textbf{Does going \emph{beyond worst-case analysis} enables to achieve a bounded convergence time in the integral case?} We show that while the worst-case complexity of the problem is infinite, its \emph{smoothed complexity} is polynomial.
\end{enumerate}
\subsubsection{A deterministic algorithm}
We give the first \emph{deterministic} algorithm with bounded convergence time for load balancing in dynamic networks. 
Our staring point is the work of~\cite{DinitzFGN17}, which gave a randomized algorithm for this problem. Not only is our algorithm deterministic, but it also has a slightly improved running time compared to the previous, \emph{randomized} algorithm. To achieve this in our algorithms a node might exchange load with two neighbors in a round, rather than one neighbor as is the case in~\cite{DinitzFGN17}.
While for some purposes, it might be necessary that each node participates in a single exchange per round, for most applications a constant number of exchanges will do (e.g.,  multiprocessor systems, virtual machine placement in data center).
As far as we are aware, our algorithm currently gives the best known guarantees on the convergence time in dynamic networks;
this includes even algorithm that do not carry the guarantee that a node exchanges loads with~$O(1)$ neighbors at a round.

Our algorithm utilizes the natural \emph{max-neighbor} balancing strategy introduced in~\cite{DinitzFGN17}, and extends upon it.
The basic idea is that in every communication round, every node attempts to connect with a neighbor that maximizes the difference in loads. For example, if node $v$ has load $w(v)$, then it tries to connect to a neighbor $u$ such that $|w(v)-w(u)|$ is \emph{maximized}.
Each node needs to find $O(1)$ neighbors to exchange load with, and this is accomplished via a simple proposal protocol: each node sends a request to establish a connection to its neighbor with the maximum difference in loads; a node that receives a request can then accept the best option and establish a connection. After such a connection is established, the two nodes balance their loads equally between them.

This algorithm, however, needs to overcome a symmetry breaking problem---a node can either send a request to a neighbor, or it can consider requests received from neighbors; it cannot do both. Otherwise, the load that it is proposing to send to its neighbor may be instead taken by the proposal it decides to accept. 
The authors of~\cite{DinitzFGN17} break this symmetry by having nodes flip a coin at every round in order to decide upon their role in this algorithm. 

In this paper, we show that actually, a node can play both roles in the algorithm \emph{simultaneously}, by ``pretending'' to be two nodes, one which only sends requests and one which only receives them. Each of the pretend nodes is given half the load for it to balance. This approach is clean and simple, removes the unnecessary randomness from the algorithm, simplifies the analysis, and even slightly improves the running time of the max-neighbor algorithm.

Consider a network consisting of $n$ nodes, where $T$ is the total amount of load in the network. Then the system achieves $\tau$-convergence (i.e., the maximum difference in load between any two nodes is at most $\tau$) in time:
\[O\left(\min \left\{n^2 \log{\frac{Tn}{\tau}}, \frac{Tn}{\tau}\right\}\right)\]
improving the results of \cite{DinitzFGN17} by a $\log n$ factor. 

\subsubsection{Integral load balancing}
When the loads model computational tasks, it is very natural to think of a basic atomic unit of load that cannot be split. To this end, we consider the same load balancing problem, with the added constraint that all loads are initially integers, and must remain so throughout the balancing process. 
This small change immediately makes the problem immensely harder. 
In the continuous case, for every graph topology and every allocation of loads that is not balanced, it is always possible to make some progress; however, this is not the case when the loads must be integral.
Consider, for example, the case of a path with $n$ nodes where the leftmost node has load $1$, the node to its right has load $2$, and so on up to load $n$ at the last node. Due to the integrality constraints, no balancing can happen between two neighboring nodes. Thus, an algorithm which only balances loads among neighbors cannot make any progress from this configuration. 

In the above example, a more clever algorithm might shift loads in a more global manner, and not only according to each node's local view.
However, this cannot be done in general:
in Section~\ref{sec: integer} we show that \emph{no} matching-based algorithm can achieve a bounded convergence time for the integral case.


As mentioned, the integral case is very natural in real-world scenarios of load balancing. 
Thus, it is somewhat disheartening that in theory we cannot guarantee any bounded convergence while load balancing systems exist and function well in the real world.
To bridge this gap between theory and practice we invoke the toolbox of smoothed analysis. Smoothed analysis was first introduced to the distributed setting in the work of Dinitz et al.~\cite{DFGN18}, who considered distributed dynamic networks. This work was later extended~\cite{MPS20} to consider different models of noise applied to the network. 
In this paper we use the mode of fractional noise with an adaptive adversary, as introduced in~\cite{MPS20}.
Roughly speaking, at each time the adversary chooses a communication graph, $\Theta(k)$ edges are being added to the graph or removed from it in a randomized manner, 
according to some \emph{noise parameter} $k\in\RR_{+}$.

Technically, the perturbation ($k$-smoothing) is done by first rounding $k$ up with probability $k - \floor{k}$, and down otherwise. 
Denote this value by $k'$.
We then choose uniformly at random a graph within Hamming distance at most $k'$ of $G_i$
(i.e., a graph that is different from $G_i$ by at most $k'$ edges), which also has some desired properties (connectivity in our case).
Note that in the typical case $\Omega (k)$ edges are perturbed, as most of the graphs with Hamming distance at most $k'$ from $G_i$ have Hamming distance $\Omega (k)$ from it. 
The fractional nature of $k$ does not change the model significantly compared to integer values of $k$ when $k>1$. 
The case $0<k<1$ is more interesting, and allows us to analyze more robust networks, where random edges do not necessarily appear in each round, but instead a single random edge appears once every $1/k$ rounds in expectation.

We present a deterministic, matching-based, load balancing algorithm that for \emph{any} $k>0$ achieves a convergence time of \[O\left( \frac{n^2}{k} \log \left( \frac{T}{\tau} \right) \log \left( n \log T \right) \right)\]
with high probability.
That is, convergence is guaranteed even with the faintest amount of noise in the system, and as $k$ increases the convergence time decreases.

A simple reduction shows that the same guarantee holds for the simpler, continuous case.
For both cases, the convergence time can get \emph{substantially lower} than possible in worst-case analysis. 
For example, when $T=\poly(n)$ there is a noise parameter satisfying  $k=\polylog(n)$, guaranteeing a convergence time of  $O(n^2/\polylog(n))=o(n^2)$, which beats the worst-case lower bound of~\cite{DinitzFGN17} for the continuous case. 
That is, in the range of $k\in[0,\polylog(n)]$, the smoothed complexity of the integral load balancing problem goes from infinity to below the best possible (in worst-case analysis) for the continuous case.

\subsection{Related work}
The rigorous study of load balancing in distributed systems was initiated by Cybenko~\cite{Cybenko89}, who studied load balancing in static networks using a diffusion strategy, where nodes average the load among neighbors.
However, he also noted that this strategy is infeasible in systems where nodes cannot balance their load (i.e., exchange work) with many neighbors in parallel. 
In static networks with hypercube topology, he suggested an improved algorithm, requiring each node to exchange load only with a single neighbor at a time.
Boillat~\cite{Boillat90} used Poisson heat equation to present a different load balancing algorithm, which he showed to converge in $O(n^2)$ steps in static networks, using Markov chains.
Following these two seminal works, a lot of research effort was put into studying load balancing in static networks; 
for example, Cybenko's technique was analyzed in specific network topologies~\cite{XuL94,XuL92,XuL93,HosseiniLMMV90}, 
and Feuilloley et al.~\cite{FeuilloleyHS15} studied load balancing with local guarantees.

Bahi et al.~\cite{BahiCV03,BahiCV05,BahiCV07} and Els\"asser et al.~\cite{ElsasserMS04} were the first to consider load balancing in dynamic networks. 
These works, alongside with a more application-oriented work of Sider and Couturier~\cite{SiderC09}, were the base for the work of Dinitz et al.~\cite{DinitzFGN17}, upon which we elaborate.
Dinitz et al.~\cite{DinitzFGN17} consider load balancing in a setting where a node can broadcast small pieces of information, such as its current load, to all of its neighbors, but can exchange load only with a single neighbor in each round (henceforth, matching-based). 
They present a randomized load balancing algorithm in this setting, which completes in $
O(\min \{ n^2 \log{\frac{Tn}{\tau}}, \frac{Tn}{\tau} \} \cdot \log n)
$ rounds with high probability (where $T$ and $\tau$ are the total load in the system and a slack parameters as defined above).
They also show that in this model, $\Omega(n^2)$ rounds are necessary for any matching-based load balancing algorithm.
Berenbrink et al.~\cite{BerenbrinkFKK19} recently analyzed load balancing in a random environment, using an algorithm different than ours.
Parts of their work might be applicable to our smoothed model, where an adversary is also present.
This might imply faster balancing can be achieved in certain cases, but their results do not apply for perfect balancing, for which case we believe our running time to be tight.

Our work continues a long line of research where \emph{a network of processors} performs the load balancing: tasks are transferred between the processors, in order to balance the amount of tasks each processor handles~\cite{PelegU86}.
A different line of work studies load balancing in \emph{client-server networks}: here, servers communicate with the clients, in order to assign clients to servers in a  balanced way.
While the scenario in our work is modeled by an arbitrary communication network, load balancing in client-server networks is modeled using bipartite graphs, and admits very different algorithmic results, both in the centralized setting~\cite{HarveyLLT06,Horn73,BrunoCS74,GalcikKS17,FakcharoenpholLN14} and in the distributed setting~\cite{AssadiBL20,BrandtKRSU20,HalldorssonKPR18,CzygrinowHSW12}.


Another related problem is that of \emph{approximate consensus}~\cite{Degroot74,DolevEA86}: here, each node has to produce an output value such that all outputs lie in a small range, and are between the minimum and maximum original values.
A classical approach for solving this problem is \emph{local averaging}~\cite{VicsekEA95}: in each round, each node adopts the average of all its neighbors values. 
Note that such an algorithm cannot be used for load balancing: not only that a node exchanges ``load'' with many neighbors at once, but more importantly, the total load is not preserved.
To see this, consider a star with an initial load  of $1$ on each leaf and $0$ on the central node.
After one round of averaging, the load on each leaf drops to $1/2$, and the total load  drops to roughly half the original one.

\emph{Smoothed analysis} was introduced by Spielman and Teng~\cite{SpielmanT09,SpielmanT04} in relation to using pivoting rules in the simplex algorithm. Since, it have received much attention in sequential algorithm design; see, e.g., the survey in~\cite{SpielmanT09}. The first application of smoothed analysis to the distributed setting is due to Dinitz et al.~\cite{DFGN18}, who apply it to the well studied problems of aggregation, flooding and hitting time in dynamic networks. In \cite{MPS20} various other models of noise were introduced for smoothing in dynamic networks, where the behavior of the flooding problem was analyzed for these new models.

\medskip
\noindent\textbf{Recent developments.}
In a recent manuscript done in parallel with our work (brief announcement in SSS 2020), Dinitz, Dolev and Kumar~\cite{DinitzD020}%
\footnote{%
The group of (Yefim) Dinitz, Dolev and Kumar~\cite{DinitzD020} only coincidentally shares an abbreviation with the group of (Michael) Dinitz, Fineman, Gilbert, and Newport, who wrote~\cite{DinitzFGN17} and~\cite{DFGN18};
to avoid confusion, we use ``Dinitz et al.'' solely for the latter.}
consider essentially the same algorithm we present in Section~\ref{sec: det-alg} for the continuous case, 
but only for static networks.
Their analysis gives a convergence time of $O\left(nD\ln(nT/\tau)\right)$, for static networks of diameter at most $D$.
Their result seems to extend to \emph{dynamic} networks, coinciding with one of our bounds for networks of arbitrary static diameter.
While we show faster convergence when the total load $T$ is small, their analysis would show faster convergence as long as at any time step the diameter is bounded by some value $D$.
For integral load balancing, they study the simplified task of local balancing: only \emph{adjacent} nodes must have similar loads.
As adjacencies change frequently in a dynamic network, such results do not directly extend to the dynamic setting, as also evident by our impossibility result in Section~\ref{sec:impossibility}.

\section{Preliminaries}

We start by formally defining the load balancing problem, and then describe the algorithm of~\cite{DinitzFGN17}, which we augment.

\medskip
\noindent\textbf{Model.} We consider a dynamic graph on $n$ nodes, which evolves over time.  This dynamic graph is represented by a sequence of simple connected graphs $\set{G_r}=\set{V, E_r}$.  Time is divided into rounds, and we refer to the duration of graph $G_r$ as \emph{round $r$}. We assume that each node learns about its current neighbors (i.e., its local view of the new topology) at the beginning of each round.

Each node $v\in V$ starts with an initial load $w_0(v)\in \RR_{\geq 0}$.  We denote by $w_r(v)$ the load of node $v$ at the end of round $r$. 
Let $T = \sum_{v\in V} w_0(v)$ be the total load in the system.  
We say that the load in the system is $\tau$-balanced at round $r$ if it holds that $\forall v,u \in V, \size{w_r(v) - w_r(u)}\leq \tau$ for some convergence parameter $\tau$.

Communication occurs locally, in a round-by-round fashion.  Within a round, there are two types of communication: exchanging information, and exchanging load. First, each node can send a small message (i.e., containing $O(\log n + \log T)$ bits) to all of its current neighbors, and receive a (small) response to that message. Second, a node can establish a connection with $O(1)$ of its neighbors.\footnote{In our case, a node will establish connections with two neighbors. In Dinitz et al.~\cite{DinitzFGN17}, a node establishes a connection with a single neighbor.}  Third, if two neighboring nodes both establish a connection with each other, then those two nodes can transfer load to each other.

In our algorithm, the nodes first exchange their load values simultaneously, and then  exchange parts of their load, also simultaneously.  In particular, we do not abuse the ability to exchange load with two neighbors in order to use a node as a relay between two of its neighbors.

\medskip
\noindent\textbf{A note on computational models.}
Most prior works on load balancing in static networks analyze the convergence time of some predefined iterative process, where nodes take decisions only based on the loads they and their neighbors currently hold.
In a more enabling setting where nodes can \emph{coordinate}, say the \textsc{congest} model, load balancing in static networks can be easily done in merely diameter $O(D)$ time---the nodes can aggregate all the load over a BFS tree and then distribute it equally.
This technique, however, might require a node to exchange load with many neighbors in some round.
To get an algorithm where nodes exchange load on a single edge in a round, one can use the token-passing technique of~\cite{HolzerW12}, and get a a linear, $O(n)$ time algorithm.
Sublinear-time convergence is impossible, as evident by the example of a star graph where all the load is initially in the central node.
In our work, we study dynamic networks, where the use of such simple global coordination mechanisms is not possible.

\medskip
\noindent\textbf{Randomized max-neighbor load balancing.}  We review the algorithm of~\cite{DinitzFGN17}, which we augment.

In each round $r \geq 1$, each node $u \in V $ flips a
fair coin to determine whether to send or receive connection proposals. 
\begin{itemize}
    \item If $u$ decides to send, it selects neighbor $v$ with a maximum difference in load, i.e., $v=\argmax_{v\in N_r(u)}\{\size{w_{r-1}(v)-w_{r-1}(u)}\}$ and sends $v$ a connection proposal. 
    \item If $u$ decides to receive, and the set  $S_r$ of nodes sending proposals to $u$ at rounds $r$ is non-empty, it accepts a proposal from a node $v$ with maximum difference of loads, i.e., $v = \argmax_{v\in S_r} \{\size{w_{r-1}(v)-w_{r-1}(u)}\}$. 
    In this case we say that $u$ and $v$ are connected in round $r$. 
\end{itemize}
In both $\argmax$ decisions, ties can be broken by
any fixed deterministic criteria. 
If two nodes $u$ and $v$ are connected in round $r$, then
they evenly balance their respective workloads by setting
$w_r(u) = w_r(v) = (w_{r-1}(u)+w_{r-1}(v))/2$.

\medskip
\noindent\textbf{Smoothing of dynamic networks.}
We use a model of smoothed analysis of dynamic networks, defined as follows.
At each step, we think of $\gtmp$ as the graph suggested by the adversary. 
The actual future graph, $\gnew$, will be a modified version of $\gtmp$, randomly chosen as a function of $\gtmp$. 
In addition, we consider the set $\gallowed$ of \emph{allowed graphs} for a specific problem. 
For load balancing, this is just the set of all connected graphs on $n$ nodes.

\begin{definition}
Let $t \in \NN$ be a parameter, $\gallowed$ a family of graphs, and $\gtmp$ a graph in $\gallowed$.
A $t$-smoothing of $\gtmp$ 
is a graph $\gnew$ which is picked uniformly at random from all the graphs of $\gallowed$ that are at Hamming distance at most $t$ from $\gtmp$.
The parameter $t$ is called the \emph{noise parameter.}
\end{definition}

We are now ready to define smoothing for dynamic graphs. The definition we use is one of a few suggested in~\cite{MPS20}, which in turn extends the original definition of~\cite{DFGN18}.

For a positive real number $x$, we define the random variable $\roundp(x)$, which takes the value $\ceil{x}$ with probability $x - \floor{x}$ (the fractional part of $x$), and $\floor{x}$ otherwise.

\begin{definition}
\label{def: smoothed_dynamic_network}
Let $k\in\RR_{+}$ be a parameter, and $\gallowed$ a family of graphs.
Let $\Hcal=(G_1,G_2,\ldots)$ be a dynamic graph, i.e., sequence of ``temporary'' graphs. 
A $k$-smoothed dynamic graph is the dynamic graph $\Hcal'=(G'_1,G'_2,\ldots)$ defined from $\Hcal$, where 
for each round $i>0$,  $G'_i$ is the $t_i$-smoothing of $G_i$, where $t_i\sim \roundp(k)$.

The definition extends to an \emph{adaptive} adversary, where $\Hcal$ and $\Hcal'$ are defined in parallel: the graph $G_i$ is chosen by an adversary after the graphs $G_1,G_2,\ldots,G_{i-1}$ and $G'_1,G'_2,\ldots,G'_{i-1}$ are already chosen, and then $G'_i$ is defined as above.
\end{definition}
The algorithm and analysis in Section~\ref{sec:smoothed analysis} deals with the harder case of an \emph{adaptive} adversary. 
\section{Deterministic load balancing}
\label{sec: det-alg}
\subsection{Our algorithm}
Our algorithm avoids the randomness used in the aforementioned method, by having every node both send and receive connection proposals at every round. Given the original communication graph $G=(V,E)$ for a round, we create a new graph $G'=(V', E')$ for that round, where $V'=\set{v_s, v_a \mid v \in V}, E'=\set{(v_s, u_a) \mid (v,u) \in E}$, i.e., splitting each node into a ``sender'' and a ``receiver'' part. Note that $G'$ is bipartite. When creating $G'$ we evenly split the load of $v$ among $v_s, v_a$, that is, $w_0(v_s)=w_0(v_a)=w_0(v)/2$.

We then run the algorithm of \cite{DinitzFGN17} on $G'$, but instead of coin flips, all $v_a$ nodes always receive and all $v_s$ nodes always send. We call this phase the \emph{interactive balancing}.

Moreover, after every balancing iteration of the algorithm of~\cite{DinitzFGN17}, we rebalance the load evenly between all $v_a, v_s$. We call this phase the \emph{internal balancing}.

\subsection{Analysis}
We start with some useful definitions.
\begin{definition}
    Given a set $V$ of nodes and a weight function $w:V \to \RR_{\geq 0}$, we define for every $u,v\in V$ the \emph{induced difference} $d_w(u,v) = \size{w(u) - w(v)}$.
\end{definition}

\begin{definition}
    Two weight functions $\wold, \wnew:V \to \RR_{\geq 0}$ are said to \emph{simulate a balancing step} over a set of nodes $V$ and a matching (a set of disjoint edges) $A = \set{(u_i,v_i)}_{i=1}^{m}$ if
    for each edge $(u_i,v_i)$:
    \[\wnew(u_i) = \wnew(v_i) = (\wold(u_i) + \wold(v_i)) / 2 \]
    and for any node not appearing in $A$, $\wnew(z) = \wold(z)$.
\end{definition}

We next define a general potential function $\phi$.
\begin{definition}
	Given a set $V$ of nodes and a weight function $w:V \to \RR_{\geq 0}$, the \emph{potential function} $\phi$ of $w$ is the sum of induced differences over all pairs of nodes:
    \[\phi(V,w) = \sum_{\set{\set{u,v}\mid u,v\in V}} d_w(u,v).\]
\end{definition}
The above sum is over all un-ordered pairs of nodes; for brevity, we later just write $u,v\in V$ in the subscript.

Recall that we write $w_r:V \to \RR_{\geq 0}$ as the weight function at the end of round $r$. 
For the set $V'$ created by splitting each node, however, we define \emph{two} weight functions for each round. The function $w_r^1$ at the beginning of round $r$, and $w_r^2$ for the end of it (after the interactive balancing).
In terms of the bipartite graph $G'$, our algorithm at iteration~$r$ can be described in two phases.
\begin{itemize}
    \item 
    Starting with weights $w_r^1$, balance according to~\cite{DinitzFGN17}, where all $v_s$ nodes send proposals, and all $v_a$ nodes receive (and answer) proposals. We end up with weights $w_r^2$.
    
    On the real graph, at this moment, $w_r(v) = w_r^2(v_s) + w_r^2(v_a)$.
    
    \item
    Starting with weights $w_r^2$, we balance over the perfect matching $(v_s,v_a)$ in order to prepare for next round. We end up with weights $w_{r+1}^1$.
    This balancing step guarantees that $w_{r+1}^1(v_s) = w_{r+1}^1(v_a) = w_{r}(v) / 2$.
\end{itemize}

We are now ready to relate our potential function defined $w$ on $V$ and the function $w^1$ defined on $V'$. 
Our first claim is that at the beginning of any round $r \geq 1$, the potential on~$V'$ relates to that on $V$ at the end of the previous round as follows.
\begin{claim}
    \label{claim:potential_connection}
     $\phi(V',w_{r}^1) = 2\phi(V,w_{r-1}).$
\end{claim}

\begin{proof}
	By a direct calculation, using the structure of $V'$:
	\begin{align*}
		\phi(V',w_{r}^1)
		=& \sum_{u, v\in V'} \size{w_{r}^1(u) - w_{r}^1(v)} \\
		=& \sum_{u,v\in V} \left( \size{w_{r}^1(u_s) - w_{r}^1(v_s)} + \size{w_{r}^1(u_s) - w_{r}^1(v_a)}\right) +\\
		& \sum_{u,v\in V} \left( \size{w_{r}^1(u_a) - w_{r}^1(v_s)} + \size{w_{r}^1(u_a) - w_{r}^1(v_a)} \right) \\
		=& \sum_{u,v\in V} 4 \size{\frac{w_{r-1}(u) - w_{r-1}(v)}{2}} 
		= 2\phi(V,w_{r-1}).\qedhere
	\end{align*}
\end{proof}

We also use the following lemma from~\cite{DinitzFGN17}, whose proof appears here for completeness.
\begin{lemma}
    \label{lem:balancing_step}
    If $\wold, \wnew:V \to \RR_{\geq 0}$ simulate a balancing step over node set $V$ and a matching $A$, then
    \begin{equation*}
        \label{ineq:balancing_step}
        \phi(V,\wnew) \leq \phi(V,\wold) - \wold(A)
    \end{equation*}
    where $\wold(A) = \sum_{(u_i,v_i)\in A} d_{\wold}(u_i,v_i)$.
\end{lemma}

\begin{proof}
	First, we show correctness for the case where the balancing step is over a single edge, i.e. $A = \set{(u,v)}$.
	In this case, we know that $\wold(A) = d_{\wold}(u,v)$, and $\wnew(u) = \wnew(v)$ so $d_{\wnew}(u,v) = 0$.
	
	For almost every node $z$ it holds that $\wold(z) = \wnew(z)$, and therefore the only different summands between the sums $\phi(V,\wold)$ and $\phi(V,\wnew)$ are the summand that corresponds to $(u,v)$, and the ones that correspond to either $(z,u)$ or $(z,v)$, for $z\neq u,v$. 
	We group the last two, and re-write:
	\begin{align*}
		\phi(V,\wnew)& - \phi(V,\wold)
		=\\
		& \left( d_{\wnew}(u,v) + \sum_{z\neq u,v} d_{\wnew}(z,u) + \sum_{z\neq u,v} d_{\wnew}(z,v) \right) - \\
		& \left( d_{\wold}(u,v) + \sum_{z\neq u,v} d_{\wold}(z,u) + \sum_{z\neq u,v} d_{\wold}(z,v) \right) \\
		=& \sum_{z\neq u,v} d_{\wnew}(z,u) + \sum_{z\neq u,v} d_{\wnew}(z,v) - \\
		& \left(\sum_{z\neq u,v} d_{\wold}(z,u) + \sum_{z\neq u,v} d_{\wold}(z,v)\right) - \wold(A).
	\end{align*}
	
	So, it is enough to prove the following inequality:
	\begin{equation}
		\label{ineq:sums}
		\sum_{z\neq u,v} \left( d_{\wnew}(z,u) + d_{\wnew}(z,v) \right) \leq \sum_{z\neq u,v} \left( d_{\wold}(z,u) + d_{\wold}(z,v) \right)
	\end{equation} 
	
	Indeed, for each $z \neq u,v$ it holds that:
	\begin{align*}
		d_{\wnew}(z,u)& + d_{\wnew}(z,v)\\
		&=\size{\wnew(z) - \wnew(u)} + \size{\wnew(z) - \wnew(v)} \\
		&= 2\size{\wold(z) - (\wold(u) + \wold(v) )/2} \\
		&= \size{2\wold(z) - \wold(u) - \wold(v)} \\
		&\leq \size{\wold(z) - (\wold(u)} + \size{\wold(z) - (\wold(v)} \\
		&= d_{\wold}(z,u) + d_{\wold}(z,v).
	\end{align*}
	Summing over $z$ proves inequality~\eqref{ineq:sums}, which proves the claim for $A=\set{(u,v)}$.
	
	For the case of multiple balancing pairs, we note that all pairs are disjoint, and so the balancing step is equivalent to a series of $m$ balancing steps, each over one edge $(u_i,v_i)$ in $A$.
	Each such step decreases the potential by $w(u_i,v_i)$, and altogether the potential decreases by $\sum_i w(u_i,v_i) = \sum_{i} d_{\wold}(u_i,v_i) = \wold(A)$.
\end{proof}

Our next claim quantifies how fast the the potential function decreases over time.

\begin{definition}
    A directed edge $(u,v)$ is said to \emph{connect} at round $r$ if at that round $u$ sends a proposal that $v$ accepts.\footnote{Specifically, $(u,v)$ and $(v,u)$ can both connect at the same round, if they both accept each other's offers.}
\end{definition}

Note that if such a connection occurs in round $r$, the amount of load shifting between $u$ and $v$ is half the difference of their weights. But actually, we only use the weight $u$ designated for sending, and $v$ for accepting. This amounts to $(1/2)d_{w_r^1}(u_s,v_a) = (1/4) d_{w_{r-1}}(u,v)$. It would be simpler to keep twice this quantity in mind, which brings us to define the following.

\begin{definition}
    \label{def: twice_shifted_load}
    We denote by $A_r$ the set of ordered edges that connect at round $r$. We use $D_r$ to denote \emph{twice} the total load shifted at round $r$:
    \[D_r := \frac{1}{2} \sum_{(u,v)\in A_r} d_{w_{r-1}}(u,v).\]
\end{definition}

\begin{lemma}
    \label{lem:potential_drop}
    For any round $r\geq 1$, we have
    \[\phi(V, w_{r}) \leq \phi(V,w_{r-1}) - D_r / 2.\]
\end{lemma}

\begin{proof}
	First, we apply Lemma~\ref{lem:balancing_step} over the set $V'$, twice for each round, showing a decrease over time. 
	Then, using Claim~\ref{claim:potential_connection}, we get the same decrease of potential in $V$.
	
	For any round $r$, we first note that in phase two (the internal balancing), the balancing step on $V'$ and the matching $A = \set{(v_s,v_a)}_{v\in V}$, is simulated by the weights $\wold = w_r^2$ and $\wnew = w_{r+1}^1$. We use the mere fact that $\wold(A) \geq 0$ (by definition), and apply Lemma~\ref{lem:balancing_step}:
	\[\phi(V', w_{r+1}^1) \leq \phi(V',w_r^2).\]
	
	The first phase of round $r$ (the interactive balancing), can be viewed as a balancing step on $V'$ and the matching $A'_r$, where $A'_r = \set{(u_s, v_a) \mid (u,v) \in A_r}$ -- meaning, all edges $(u,v)$ such that $u$ made a proposal and $v$ accepted it at round $r$.
	The reason $A'_r$ is a matching over $V'$, is that in our algorithm no sender part $v_s$, and no answering part $v_a$, can connect to two nodes simultaneously. 
	This balancing step is simulated by the weights $\wold = w_r^1$ and $\wnew = w_r^2$, which means
	\[\phi(V', w_r^2) \leq \phi(V',w_r^1) - w_r^1(A'_r) \]
	We relate this weight to $D_r$ as follows:
	\[ w_r^1(A'_r) = \sum_{(u_s,v_a)\in A'_r} d_{w_r^1}(u_s,v_a) = \frac{1}{2}\sum_{(u,v)\in A_r} d_{w_{r-1}}(u,v) = D_r ,\]
	and conclude that
	\[\phi(V', w_r^2) = \phi(V',w_r^1) - D_r .\]
	
	Combining the two phases, we get
	\[\phi(V', w_{r+1}^1) \leq \phi(V',w_r^1) - D_r ,\]
	and translating back to the potential over $V$, we get
	\[\phi(V, w_{r}) \leq \phi(V,w_{r-1}) - D_r / 2.\qedhere \]
\end{proof}

We next show that in our deterministic algorithm, at every iteration --- \emph{every} edge must have some weight shifted nearby that relates to the difference on this edge. Formally, we show the following:
\begin{claim}
    \label{claim:covering_edge}
    For each edge $e = (u,v)\in E_r$, there exists another edge $e' \in E_r$ that connects at round $r$ and fulfills:
    \begin{itemize}
        \item 
        $e'$ is within $3$ hops from $e$ in the graph $G_r$.
        
        \item
        The weight difference of $e'$ exceeds that of $e$ in round $r-1$, meaning: $d_{w_{r-1}}(e') \geq d_{w_{r-1}}(e)$.
    \end{itemize}
    We call $e'$ a \emph{covering edge} for $e$.
\end{claim}

The proof is by exhaustive case analysis, finding a covering edge nearby.

\begin{proof}
	We know that at round $r$, all nodes sent proposals, and any node that got at least one proposal chose the best one. We exhaust all cases:
	\begin{itemize}
		\item If $u$ sent a proposal to $v$, and $v$ accepted it, then $(u,v)$ connects at round $r$ and serves as a covering edge for itself.
		
		\item If $u$ sent a proposal to $v$ who then accepted a different proposal from $z \neq u$, then it must be the case that $d_{w_{r-1}}(z,v) \geq d_{w_{r-1}}(u,v)$, and also $(z,v)$ connects and it is within $3$ hops, so it covers $(u,v)$.
		
		\item If $u$ sent a proposal to some other node $z \neq v$ who then accepted, then we know $d_{w_{r-1}}(u,z) \geq d_{w_{r-1}}(u,v)$, and also $(u,z)$ connects and it is within $3$ hops, so it covers $(u,v)$.
		
		\item If $u$ sent a proposal to $z \neq v$, and $z$ accepted a different proposal from $y\neq u$, then we have $d_{w_{r-1}}(y,z) \geq d_{w_{r-1}}(u,z) \geq d_{w_{r-1}}(u,v)$, and also $(y,z)$ connects and it is within $3$ hops, so it covers $(u,v)$.
	\end{itemize}
	In all possible cases, we have a covering edge, which concludes the proof.
\end{proof}

We now turn to prove a lower bound for the amount of load shifted at round $r$. This statement is stronger than the one in~\cite{DinitzFGN17} in two ways: (i) It holds with probability $1$; (ii) We gain a multiplicative factor of $\Theta(\log n)$.

\begin{lemma}
    \label{lem:load_shifted_lb}
    For any round $r$, we have
    \[D_r \geq t_{r-1} / 30 ,\]
    where $t_r := \max_{u,v \in V} d_{w_r}(u,v)$ is the maximal difference between any two nodes in $V$ (not necessarily connected), at the end of round $r$.
\end{lemma}

\begin{proof}
	At each round $r$, we focus on two nodes that had maximal weight difference by the end of the previous round $r-1$.
	
	Formally, we take $u_{\min}, u_{\max}$ such that $d_{w_{r-1}}(u_{\min},u_{\max}) = t_{r-1}$. We now focus on the current graph $G_r = (V,E_r)$, and specifically on a \emph{shortest path} $P = (e_1,\dots,e_k)$ (in terms of hops) from $u_{\min}$ to $u_{\max}$ in the graph $G_r$.
	
	We use the following claims from the original analysis in~\cite{DinitzFGN17}:
	\begin{enumerate}
		\item
		Using the triangle inequality, the difference over all path edges, are larger than the difference between the two ends of the path:
		\[\sum_{e \in P} d_{w_{r-1}}(e) \geq t_{r-1} .\]
		
		
		\item
		If we greedily choose into a set $P'$ the heaviest available edge of $P$ (according to $d_{w_{r-1}}(e)$), and with each choice $e$ turn all edges in its $6$-neighborhood unavailable --- we end up choosing a subset $P'$ of edges of $P$ that has a significant weight:
		\[\sum_{e\in P'} d_{w_{r-1}}(e) \geq (1/15)\sum_{e\in P} d_{w_{r-1}}(e) .\]
	\end{enumerate}
	
	Now, for any edge $e \in E_r$, let $c(e)$ be the edge that covers it in $G_r$, as chosen by the exhaustive analysis in the proof of Claim~\ref{claim:covering_edge}.
	
	We show that the edges of the set $P'$ do not share any covering edges. Indeed, by the definition of $P'$, every two edges $e, e'$ are at least $7$ hops apart. If we assume they have the same covering edge $c(e) = c(e')$, then through it there is a path of length $6$ from $e$ to $e'$, contradicting the fact that $P$ is a shortest path.
	
	The last piece of the puzzle is that in round $r$, on any covering edge $c(e) = (y(e),z(e))$, the load we actually shift is between $y(e)_s$ and $z(e)_a$.
	
	Now, we invoke the two quoted claims along with our last observations, and get
	
	\begin{align*}
		D_r = (1/2) \sum_{(u,v)\in A_r} d_{w_{r-1}}(u,v)
		&\geq (1/2)\sum_{e \in P'} d_{w_{r-1}}(c(e)) \\
		&\geq (1/2)\sum_{e \in P'} d_{w_{r-1}}(e) \\
		&\geq (1/30)\sum_{e \in P} d_{w_{r-1}}(e)\\
		&\geq t_{r-1} / 30
	\end{align*}
	which holds with probability $1$, concluding the proof.
\end{proof}

Combining Lemma~\ref{lem:potential_drop} with Lemma~\ref{lem:load_shifted_lb}, we obtain the next corollary.
\begin{corollary}
\label{cor:potential_decrease}
For every round $r \geq 1$, it holds that:
\[\phi(V, w_{r}) \leq \phi(V,w_{r-1}) - t_{r-1}/60 .\]
\end{corollary}

The last part is deducing convergence bounds, as follows.
\begin{theorem}
    \label{thm:main}
    There is a deterministic load-balancing algorithm that guarantees $\tau$-convergence in a distributed dynamic network with total load~$T$ within $O(\min\set{n^2 \log\left( nT/\tau \right), nT/\tau}$ rounds.
\end{theorem}

The proof follows the same lines as the one in~\cite{DinitzFGN17}, with our improved parameters.
We present it here for completeness.
\begin{proof}
    We run our algorithm for $r_0 = \min\set{60n^2 \ln\left( nT/\tau \right), 60nT/\tau}$ rounds.
    To bound the convergence time, we first bound the initial value of the potential function at time $r=0$, and then show how much it decreases at each round.
    Let $w_0$ be the weight function on $V$ at the beginning of the process; so, $\phi(V,w_0) \leq nT$:
    \begin{align*}
        \phi(V,w_0) &= \sum_{u,v \in V} d_{w_0}(u,v)
        \leq \sum_{u,v \in V} \max\set{w_0(u),w_0(v)}\\
        &\leq \sum_{u \in V} \sum_{v\in V} w_0(v)
        = \sum_{u \in V} T = nT
    \end{align*}
where the second inequality is because we sum over all $u,v$ instead of just pairs where $w_0(u) \leq w_0(v)$.

We are now ready to bound the convergence time.
\medskip\\
\textbf{First bound (multiplicative decrease)}
We show that the potential function decreases multiplicatively at each round.

At each round $i \leq r$, the maximum difference in the graph $G_i$ is $t_i$. Therefore, $\phi(V,w_i) \leq n^2 t_i$, or alternatively $t_i \geq \phi(V,w_i) / n^2$. Combining this with Corollary~\ref{cor:potential_decrease}, we get
\begin{align*}
\phi(V,w_{i+1}) &\leq \phi(V,w_i) - t_i/60 \\ &\leq \phi(V,w_i) - \phi(V,w_i) / (60n^2) \\
&= \phi(V,w_i) \left( 1-1/(60n^2) \right)
\end{align*}
which proves the multiplicative decrease. Applying this bound repeatedly, we get:
\[ \phi(V,w_r) \leq \phi(V,w_0)\left( 1-1/(60n^2) \right)^r \leq nT\left( 1-1/(60n^2) \right)^r.\]

Next, we note that if at some round $r$ we have $\phi(V,w_r) \leq \tau$, then by monotonicity of $\phi$ over time, we have $\phi(V,w_i) \leq \tau$ for any $i \geq r$, which means that as of round $r$, the system stays $\tau$-converged.
Hence, it enough to find $r$ such that 
\[nT\left( 1-1/(60n^2) \right)^r \leq \tau.\]
Specifically, this holds for $r_0 = 60n^2 \ln\left( nT/\tau \right)$, since
\[nT( 1-1/(60n^2) )^{r_0} \leq nT\cdot e^{-\frac{r_0}{60n^2}} = nT\cdot e^{-\ln (nT/\tau)}  = \tau.\]
Overall, starting round $r_0 = 60n^2 \ln\left( nT/\tau \right)$, the network stays $\tau$-converged.
\medskip\\
\textbf{Second bound (additive decrease)}
We are guaranteed that \\$\phi(V,w_0) \leq nT$. Assume that at round $r$ we are not yet converged. That is, $t_r \geq \tau$, and in particular $\phi(V,w_r) \geq n\tau/2$. Indeed, by taking $u,v$ with maximal different, and summing over a subset of the pairs (assuming $n\geq 4$):
\begin{align*}
    \phi(V,w_r)
    &\geq \sum_{z\neq u,v} \size{w_r(z) - w_r(u)} + \size{w_r(z) - w_r(v)} \\
    &\geq \sum_{z\neq u,v} \size{w_r(u) - w_r(v)}
    = (n-2)t_r
    \geq n\tau/2.
\end{align*}
On other hand, we are guaranteed that at each previous round $i < r$, the potential dropped by $t_{i}/60 \geq \tau/60$, and so
\[n\tau/2 \leq \phi(V,w_r) \leq \phi(V,w_0) - r\cdot\tau/60 \leq nT - r\cdot\tau/60,\]
or put differently, $r \leq (60n/\tau) \cdot (T-\tau/2) \leq 60nT/\tau$.

Hence, at any round after $r_0 = 60nT/\tau$, the network must be $\tau$-converged.
\end{proof}

\section{Integral load balancing}
\label{sec: integer}

When considering loads as computational tasks, it is very natural to think of a basic atomic unit of load that cannot be split. This can be modeled by requiring that the load on each node is integral at all times, and a balancing step done on an edge can only balance its nodes up to an integer.

Our goal is to reach a $\tau$-converged network, for a parameter $\tau$; without loss of generality, we assume hereinafter that $\tau$ is an integer.
Unlike its counterpart, where $\tau$ can be arbitrarily small, in the integral case achieving $1$-convergence is the the best one can hope for. 
This is clear if $T$ is not divisible by $n$.
When $T$ is divisible by $n$, $0$-convergence is possible. However, it is equivalent to $1$-convergence: in any non $0$-converged configuration there must be a lighter load and a heavier load, leading to a gap of at least $2$.

In Section~\ref{sec:impossibility} we show that greedily performing load balancing is a flawed strategy in the integral case. 
Specifically, two nodes with gap $1$ cannot balance each other any further.
In section~\ref{sec:smoothed analysis} we use smoothed analysis to show that the lower bound is fragile: if only a small amount of noise is added to the network, greedy load balancing is possible again.

\subsection{Impossibility of integral load balancing}
\label{sec:impossibility}
When all loads must be integral, an adaptive adversary can prevent any matching-based load balancing algorithm from converging. 

\begin{theorem}
\label{thm:impossibility}
There exists a dynamic graph on $n$ nodes defined by an adaptive adversary, 
such that no matching-based load balancing algorithm
can reduce the maximal difference between loads in the graph below $n-1$, 
regardless of the running time.
\end{theorem}

The rest of this section is dedicated to the proof of the theorem. 
Consider a line on $n$ nodes, where node $i$ starts with load $i$.
Throughout the process, the adversary keeps the graph in a line structure, but may swap pairs of adjacent nodes in it.

We define a set of arrays $\set{a_t}_{t\in\NN}$, each of length $n$,
where $a_t[i]$ is the load of the $i$-th node in the line at time $t$. 
The initial weight assignment is $a_1[i]=i$ for all $i$.
Note that the identity of the $i$-th node in the line may change over time, due to the topology changes made by the adversary, so $a_t[i], a_{t'}[i]$ might represent loads of different nodes.

Consider a matching-based algorithm, and an adversary that works as follows: after a pair connects, the adversary orders the pair in ascending order, swapping the nodes if necessary.
For example, if two adjacent node had weights $(3,8)$ and the algorithm chose to re-balance them to $(7,4)$, then the adversary will swap the nodes to get $(4,7)$. In terms of the corresponding loads arrays, we will have $a_t[i] = 3,a_t[i+1] = 8$, and after the balance and swap, $a_{t+1}[i] = 4, a_{t+1}[i+1] = 7$. 

Our proof relies on the following fact, which holds for every index $i$ of the initial array $a_1$.
\begin{align}
    \label{ineq:integral_lb_cond}
    a_1[i] \leq  a_1[i+1] \leq a_1[i]+1 .
\end{align}
Using this condition, we show that any matching-based balancing algorithm cannot decrease the maximal difference in the network, compared to the initial state. 
More formally, we prove that $a_t[1] \leq a_1[1]$ and $a_t[n] \geq a_1[n]$, for any $t \geq 1$, by proving a more general claim.

\begin{definition}
    We denote by $p_t^k = \sum_{j=1}^{k} a_t[j]$ the sum of the first $k$ cells at time $t$. 
    For $k=0$, we denote $p_t^0 = 0$.
\end{definition}

A single cell of $a_t$ can now be expressed as a difference of two consecutive prefixes: $a_t[k] = p_t^k - p_t^{k-1}$. 
We also use the prefixes to express the change in the value of a single cell through time:
\begin{equation} 
\label{eq:cell_through_time}
\begin{split}
    &a_t[k] - a_1[k] = \\
    &\left(p_t^k - p_t^{k-1}\right) -\left(p_1^k - p_1^{k-1}\right) = \left( p_t^k - p_1^k \right) -\left( p_t^{k-1} - p_1^{k-1} \right).
\end{split}
\end{equation}

We are now ready to prove the following.
\begin{lemma}
    \label{lem:integral_prefixes}
    For all $k$ and $t$, the sum of the $k$-prefix (the first $k$ cells) at the $t$-th array, is smaller than the corresponding sum in the initial array. That is, 
    $p_t^k \leq p_1^k.$
\end{lemma}

\begin{proof}
    The proof follows by induction on $t$, proving simultaneously for all values of $k$. The induction base is trivial as we use $t=1$, for which there is an equality.
    
    We go on to show the induction step. Fix iteration $t$, and assume the claim holds for iterations $1,\dots, t-1$. In particular, by induction hypothesis,  we assume the inequality
    \[ p_{t-1}^k - p_1^k \leq 0\]
    holds for any value $k\in\set{1,\dots,n}$.

    
    We are now ready to show that at time $t$ the induction statement holds for all values of $k$ as well.
    First, note that if the $k$-prefix at iteration $t$ stayed the same as in iteration $t-1$, this means that $p_t^k = p_{t-1}^k$, which means we have nothing to prove, as by the hypothesis 
    \[ p_t^k - p_1^k = p_{t-1}^k - p_1^k \leq 0.\]
    
    From now on, we assume that at time $t$ the sum over the $k$-prefix indeed changed from time $t-1$. This means the nodes in cells $a_{t-1}[k]$ and $a_{t-1}[k+1]$ connected and exchanged loads at time~$t$. As the adversary only swaps adjacent nodes, the total load of the $k-1$-prefix necessarily stayed the same in this case, meaning $p_t^{k-1} = p_{t-1}^{k-1}$. Thus, the difference only depends on the $k$-th cell: $p_t^k - p_{t-1}^k = a_t[k] - a_{t-1}[k]$.
    
    Since the cells $k$ and $k+1$ of $a_{t-1}$ represent nodes that connect at round $t$, we know they did not exchange weights with other nodes in this round, and their total weight remained the same: 
    $a_{t}[k] + a_{t}[k+1]=a_{t-1}[k] + a_{t-1}[k+1]$.
    After the adversarial switch, their weights will be ordered increasingly: $a_t[k] = \min\set{a_{t}[k], a_{t}[k+1]} \leq \floor{\left(a_{t-1}[k] + a_{t-1}[k+1]\right)/2}$, and so
    \[ p_t^k - p_{t-1}^k = a_t[k] - a_{t-1}[k] \leq \floor{\frac{a_{t-1}[k+1] - a_{t-1}[k]}{2}} ,\]
    where the last inequality uses the fact that $a_{t-1}[k]$ is an integer and can be pushed inside the round-down function.
    
    Now, we use Equality~\eqref{eq:cell_through_time} and the induction hypothesis for time $t-1$ (with $k-1$ and $k+1$), to bound the difference from the above expression:
    \begin{align*}
      a_{t-1}[k+1] - a_{t-1}[k] 
      \leq& a_1[k+1] +\left( p_{t-1}^{k+1} - p_1^{k+1} \right) -\left( p_{t-1}^{k} - p_1^{k} \right)-\\
      & a_1[k] - \left( p_{t-1}^k - p_1^k \right) + \left( p_{t-1}^{k-1} - p_1^{k-1} \right)\\
      \leq& a_1[k+1] - a_[k] - 2\left(p_{t-1}^k - p_1^k\right)\\
      \leq& 1 - 2\left(p_{t-1}^k - p_1^k\right)
    \end{align*}
    where the last inequality uses Inequality~\eqref{ineq:integral_lb_cond} for the array $a_1$.
    
    Combining these we deduce the $k$-prefix cannot decrease at iteration $t$ more than it increased in the previous $t-1$ iterations. Indeed,
    \[ p_t^k - p_{t-1}^k \leq \floor{1/2 - \left(p_{t-1}^k - p_1^k\right)} = \floor{1/2 + p_1^k - p_{t-1}^k)} = p_1^k - p_{t-1}^k\]
    where the last equality follows from the fact that  $\left( p_1^k - p_{t-1}^k\right)$ is an integer, and adding 1/2 and rounding it down leaves it as it was.
    
    Canceling the term $p_{t-1}^k$ concludes the induction step, and the proof.
\end{proof}

Using Lemma~\ref{lem:integral_prefixes} with $k=1$ we get $a_t[1] \leq a_1[1]$ for any time $t$. Also, we recall that at any time $t$, the total load in the network, $T$, stays fixed, so we always have $p_t^n = T$. We can now apply the lemma with $k=n-1$ to get
\[a_t[n] - a_1[n] = p_t^n - p_t^{n-1} - p_1^n + p_1^{n-1} = T - T + \left(p_1^{n-1} - p_t^{n-1}\right) \geq 0.\]
Hence, $a_t[n] \geq a_1[n]$. And so, for any iteration $t \geq 1$ we have a maximal difference at least $a_t[n] - a_t[1]$, which is in turn larger than $a_1[n] - a_1[1] \geq n-1$,
as claimed in Theorem~\ref{thm:impossibility}.

\subsection{Beyond worst-case analysis}
\label{sec:smoothed analysis}
In this section, we consider a way to circumvent the impossibility result proved in the previous section. 
To this end, we consider the \emph{smoothed complexity} of the problem, where the analysis is done over a noisy network.

\medskip
\noindent\textbf{Noisy networks.}
By considering the smoothed complexity of the problem, we show that the impossibility result is actually very fragile.
That is, if the adversary does not entirely control the network, but rather some noise might be added to the topology changes at each iteration, then there is an algorithm that guarantees convergence even in the integral case.

We adopt the smoothing framework of~\cite{DFGN18} and~\cite{MPS20}: given a noise parameter $k > 0$ (not necessarily an integer), at each round roughly $k$ potential edges of the current graph are flipped (either added or deleted).
We use the following lemma to analyze the probability of a helpful edge being added by this process.
It appears in~\cite{DFGN18} for an integral $k$, and implicitly extended in~\cite{MPS20} to handle all values $0<k$ (specifically, $0<k<1$).

\begin{lemma}
\label{lem: LB hitting a set of edges}
    There exists a constant $c_1>0$ such that the following holds.
    Let $\gtmp \in \gallowed$ be a graph, and $\emptyset \neq S \subseteq \binom{[n]}{2}$ a set of potential edges.
    Let $0 < k \leq n/16$ be the noise parameter, such that $k \size{S} \leq n^2/2$.
    
    If $\gnew$ is a $k$-smoothing of $\gtmp$, then the probability that $\gnew$ contains at least one edge from $S$ is at least $c_1 k\size{S} / n^2$.
\end{lemma}


Our algorithm iteratively reduces the gap between the current maximal and minimal loads, until reaching the desired gap, $\tau$.
In each iteration, only nodes with very large load (heavy nodes) and very small load (light nodes) perform balancing steps. 
After a sufficient number of balancing steps, all lightest (or heaviest) loads are eliminated, resulting in a guaranteed reduction in the gap between the size of the largest and smallest loads in the system.

Such an algorithm would fail miserably without the presence of noise, as an adaptive adversary can always disconnect light and heavy nodes---see Section~\ref{sec:impossibility}. 
However, in a noisy network, this guarantees the ideal $1$-convergence is reached within $\tilde{\Theta}\left(n^2 \ln T/k\right)$ rounds. 
Specifically, only a single random edge per round already gives a round complexity of $\tilde{\Theta}\left(n^2 \ln T\right)$, while with no noise the problem is unsolvable.


The following algorithm, \alggapreduce{}, 
encapsulates a single iteration of gap reduction in the network.
We use $n$ for the number of nodes, $T$ as the total load in the network, $\tau$ as the desired maximal gap, and $k$ as the noise parameter; 
all these parameters are global and do not change throughout the multiple executions of the algorithm, and we assume all the nodes know them.\footnote{If they are not known to all nodes, we can use standard techniques (e.g. flooding) in order to compute and disseminate them at the beginning, with no adverse effect on our asymptotic convergence time.}
In each call to the algorithm, we use $m$ to denote the minimal load and $M$ the maximal one.
These values are flooded in the beginning of the algorithm by standard routines of minflood and maxflood---each node just re-sends the minimal and maximal values it has seen so far, for $n$ rounds (see, e.g.,~\cite{Hromkovic96,HedetniemiHL88,FraigniaudL94}).
Let $\psi=M-m$; we call a node $v$ \emph{light} if $w(v) < m+\psi/4$, \emph{heavy} if $w(u) > M-\psi/4$, and \emph{balanced} otherwise.
From $m$ and $M$, the nodes compute $\psi$ and also deduce what are the weight ranges for light and heavy nodes---all of which do not change later in the execution.
In the appendix, we show that the algorithm can be adjusted for weaker models where retrieving $M, m$ and $\psi$ is impossible.

In the rest of the algorithm, the nodes repeat a propose-accept routine, as follows.
Each light node sends a proposal to a heavy neighbor if it has one.
A heavy node $v$ that receives proposals accepts one from some node $u$, 
and they average their weights:
the weight of $u$ is updated to 
$w_{\new}(u)\gets\floor{\frac{w(u) + w(v)}{2}}$, 
and the weight of $v$ is updated to 
$w_{\new}(v)\gets\ceil {\frac{w(u) + w(v)}{2}}$.


\begin{algorithm}[h]
\caption{Algorithm \alggapreduce{} at node $u$}
\label{algo: improve convergence}
\begin{algorithmic}[1]
\State \textbf{run} minflood and maxflood to retrieve $m$ and $M$ 
\Comment{$n$~rounds}
\label{line: gapreduce initial flooding}
\State$\psi \gets M - m$; \textbf{define} \emph{heavy} and \emph{light}

\For{$r_0 = 5n^2\ln (n\ln T)/(c_1 k)$ rounds}
\label{line: gapreduce main loop}
    \If{$u$ is light}
        \State $v' \gets \argmax\{w(v)\mid v\in N(u)\}$
        \If{$v'$ is heavy} 
             \textbf{propose} to $v'$
        \EndIf
    \EndIf
    \If{$u$ is heavy and got a proposal}
        \State $v' \gets \argmin\{w(v)\mid v$ proposed to $u\}$
        \State\textbf{accept} proposal from $v'$
    \EndIf
\EndFor
\end{algorithmic}
\end{algorithm}

The algorithm contracts the gap between the lightest and heaviest node by a $3/4$ factor, as follows.
\begin{lemma}
    \label{lem: alg contract reduces by 3/4}
	If \alggapreduce{} is performed on a $\psi$-converged dynamic network, $\psi\geq2$, 
    then it outputs a $(3\psi/4)$-converged network with probability at least $1-1/(n\ln^2 T)$.
\end{lemma}
Note that we deal with the integral case, so a $\psi$-converged network is automatically $\floor{\psi}$-converged.
Specifically, the claim shows that \alggapreduce{} turns a $3$-converged network into a $2$-converged one, and a $2$-converged network into a $1$-converged network.

\begin{proof}
    In Line~\ref{line: gapreduce initial flooding}, each node repeatedly forewords the minimal and maximal load values it is aware of.
    In each round the communication graph is connected, so at least one new node learns $m$ and one learns $M$.
    Thus, after $n$ rounds all nodes know the true values of $m$ and $M$, by which they deduce  $\psi$ and can distinguish light and heavy nodes.    
    
    In the main loop (Line~\ref{line: gapreduce main loop}) balancing occurs. 
    At any round in which there is at least one edge between any light node and any heavy node, a proposal is sent and an integral balancing step must occur.
    Moreover, \emph{any} balancing step done in the algorithm is of this form---between some light node $u$ and some heavy node $v$.
	After this step, it is crucial that both nodes become balanced.
	For~$\psi \geq 4$:
    \[
        w_{\new}(u)=
        \floor{\frac{w(u) + w(v)}{2}}
        \geq \frac{w(u) + w(v)}{2} - \frac{1}{2} 
        >    m + \frac{3}{8}\psi - \frac{1}{2}
        \geq m + \frac{\psi}{4} ,
    \]
    were we used $w(u)\geq m$ and $w(v)> m+3\psi/4$ for the central inequality, and $\psi \geq 4$ for the last inequality. Similarly, 
    \[
        w_{\new}(v)=
        \ceil{\frac{w(u) + w(v)}{2}}
        \leq \frac{w(u) + w(v)}{2} + \frac{1}{2}  
        < M - \frac{3}{8}\psi + \frac{1}{2}
        \leq M - \frac{\psi}{4} ,
    \]
    using $w(u) < M - 3\psi/4$, $w(v) \leq M$ and $\psi \geq 4$.
    The fact that $w_{\new}(u)\leq w_{\new}(v)$ shows that both $u$ and $v$ are now balanced.
    For $\psi=3$, two nodes of loads $m, m+3$ exchange one unit, to have $w_{\new}(u)=m+1$ and $w_{\new}(v)=m+2$, both are balanced.
    For $\psi=2$, both end up with load $m+1$ and become balanced.
    
    At any round $r$ of the main loop, denote by $L_r\subseteq V$ the set of light nodes at the beginning of the round, and by $H_r\subseteq V$ the set of heavy nodes at the beginning of the round.
    Note that in a single execution of \alggapreduce{}, the definition of light and heavy does not change. 
    The sets of light and heavy nodes, however, do change.
    We define the potential function $f$ at round $r$ to be
    \[
     f(r) := \min\set{\size{H_r}, \size{L_r}} .
    \]
    In fact, if $\size{H_0}< \size{L_0}$ then the potential is $\size{H_r}$ in each round $r$, and otherwise it is always $\size{L_r}$, since each balancing step removes exactly one node from each set.
    
    Once we reach round $r'$ where $f(r')$ drops to $0$, it must be the case that $H_{r'} = \emptyset$ or $L_{r'} = \emptyset$, i.e., either no node of weight less than $m + \psi / 4$ is left, or no node of weight more than $M - \psi / 4$ is left. In both cases, the largest gap between two nodes must now be at most $3\psi/4$. 
    It is left to show that with high probability the potential drops to $0$ within $r_0=5n^2\ln (n\ln T)/(c_1 k)$ rounds.

    From here on, we break the loop into $f$ phases, enumerated from $j=f$ down to $j=1$, such that by the end of phase $j$, the potential is less than $j$ w.h.p.
    Phase $j$ will consist of $t_j$~rounds.

    Fix $1\leq j\leq f$. 
    Phase $j$ starts at round $\sum_{i=j+1}^{f}t_i$;
    if by this round the potential is already below $j$, we are done. 
    Otherwise, consider a round $r$ in this phase, starting with potential $f(r)=j$.
    Over the noise of the network, the probability of the smoothed graph $G'_r$ to have an edge between a light and a heavy node can be bounded from below by using Lemma~\ref{lem: LB hitting a set of edges}.
    Adding an edge from $L_r\times H_r$ will decrease the potential.
    Since $\min\set{\size{L_r},\size{H_r}}=j$, there are at least~$j^2$ such edges, which is enough to guarantee fast progress.
    
    \medskip
\noindent\textbf{Fast phases.} 
    If $j \geq n/\sqrt{2k}$, there are so many potentially good edges that the lemma has to be applied with a much smaller subset. 
    We apply the lemma with a set $S \subseteq L_r\times H_r$, of size $\size{S} = n^2/(2k)$.
    Hence, with probability at least $c_1 k \size{S}/n^2 = c_1/2$, $G'_r$ includes a good edge, and a balancing step is guaranteed.
    Thus, with probability at least $c_1/2$ the potential is decreased by at least $1$.
    After $t_{j} = 4\ln (n \ln T) / c_1$ rounds, the probability of not performing any balancing step is at most
    \[\left(1 - \frac{c_1}{2}\right)^{t_{j}} \leq e^{-2\ln (n \ln T)} \leq \frac{1}{(n \ln T)^2} \;.
    \]

    \medskip
\noindent\textbf{Slow phases.}
    If $j \leq n/\sqrt{2k}$, we can apply the lemma with a subset $S \subseteq L_r \times H_r$ of size $\size{S} = j^2$.
    Thus, with probability at least $c_1 k  j^2 / n^2$ an integral balancing step occurs and the potential drops by $1$.    
    After $t_{j} = 2n^2 \ln (n \ln T) / (c_1 k j^2)$ rounds, the probability of not performing any balancing step is at most
    \[\left(1-\frac{c_1 k j^2}{n^2}\right)^{t_j} \leq e^{-2\ln(n \ln T)} \leq \frac{1}{(n \ln T)^2}.
    \]

    \medskip
\noindent\textbf{Combining all phases.} Since the potential can only decrease, it is enough to run the algorithm through all phases. Each phase starts with potential $i$ and ends with potential at most $i-1$, with probability at least $1 - 1/(n \ln T)^2$.
    By definition, the initial potential cannot be too large: $f \leq n$. Thus, after running all phases, by union bound, the probability of the potential not reaching $0$ (the \emph{error probability}) is at most 
    $$\frac{f}{(n \ln T)^2} \leq \frac{1}{n \ln^2 T}$$
    
    
    We next analyze the round complexity of all phases together. 
    Each fast phase requires $4\ln (n \ln T) / c_1$ rounds. 
    Since $f \leq n$, they add up to at most $4n\ln (n \ln T) / c_1$, which in turn can be bounded from above by 
    \[\frac{n^2\ln (n \ln T)}{c_1 k},\]
    as $k \leq n/16$.
    
    For the slow rounds, starting with potential $f_{\slow} = \floor{n/\sqrt{2k}}$ and all the way down to potential $1$ (the last phase), we get:
    \[
    \sum_{j=1}^{f_{\slow}} t_j =
    \sum_{j=1}^{f_{\slow}} \frac{2n^2 \ln (n \ln T)}{c_1 k j^2} = \frac{2n^2 \ln (n \ln T)}{c_1 k} \sum_{j=1}^{f_{\slow}} \frac{1}{j^2} \leq \frac{4n^2 \ln (n \ln T)}{c_1 k} ,
    \]
    where the last step uses the known solution to the Basel problem:
    \[
    \sum_{j=1}^{f_{\slow}} \frac{1}{j^2} \leq \sum_{j=1}^{\infty} \frac{1}{j^2} = \frac{\pi^2}{6} < 2 .
    \]
    
    Adding up the two terms, we get a total round complexity of $5n^2\ln (n \ln T) / (c_1 k)$.
\end{proof}



By repeatedly executing \alggapreduce{} we can easily make the network $\tau$-converged.

\begin{theorem}
\label{thm:smoothing_main}
    In a $k$-smoothed network with total load $T$, integral load balancing up to $\tau$-convergence can be done in 
    \[
    O\left(\frac{n^2\ln (T/\tau)\ln (n\ln T)}{k}\right)
    \] 
    time w.h.p.
\end{theorem}

\begin{proof}
Given a network as in the statement, we perform
$ 4 \ln (T/\tau)$ consecutive calls to \alggapreduce{}.
Starting with a gap that might be as high as $T$, after $t$ successful calls the gap is guaranteed to be at most $(3/4)^t T$.
Hence, $4 \ln (T/\tau)$ successful executions are enough to guarantee $\tau$-convergence.

Each call to \alggapreduce{} has error probability of at most $1/(n \ln^2 T)$, and a union bound shows that the total error probability is smaller than $1/n$.
\end{proof}

\medskip
\noindent\textbf{The continuous case.}
Continuous load balancing is strictly easier than integral load balancing. 
In the continuous case, we do not have a guarantee of integrality of the initial weights, so running an algorithm designed for the integral case in a continuous one is not entirely trivial. 
One could adjust and simplify our algorithm for the continuous case, but here we take a more general path and outline simple reduction from the integral case to the continuous one.

Given a continuous balancing problem with goal $\tau$ and total load $T$, we set a new base unit of $\tau/2$, and only balance over multiples of it. 
For each node $v$ we break its weight to multiples of $\tau/2$ and a remainder, $w(v) = w'(v)\cdot \tau/2 + r(V)$, such that $0 \leq r(v) < \tau/2$ (these value are unique).
We run integral balancing over $w'(\cdot)$, aiming for $1$-convergence, disregarding the remainders. 
In terms of base units, our total load is now $T'$ satisfying $2T/\tau - n \leq T' \leq 2T/\tau$.

After the execution each node is left with the new balanced weight $w_{\new}'(v)$, and the remainder left aside: $w_{\new}(v) = w_{\new}'(v)\cdot\tau/2 + r(v)$. 
Thus, the network is now $\tau$-converged:
\begin{align*}
&\size{w_{\new}(u)-w_{\new}(v)}\\
&= \size{w_{\new}'(u)\cdot\tau/2 + r(u) - w_{\new}'(v)\cdot\tau/2 - r(v)}\\
&\leq \size{w_{\new}'(u) - w_{\new}'(v)}\tau/2 + \size{r(u)-r(v)} \leq \tau/2 + \tau/2
= \tau
\end{align*}

The run of the integral algorithm uses the same $k, n$, and new parameters $\tau' = 1$, $T' = O(T/\tau)$, for which $T'/\tau' = O(T/\tau)$. Thus, Theorem~\ref{thm:smoothing_main} applies to the continuous setting as well.

\medskip
\noindent\textbf{A greedier approach?}
It is quite peculiar that at each iteration, \alggapreduce{} only balances very light nodes with very heavy nodes, but ignores completely other (more moderate) balancing options.
Although this might seem wasteful, we give some intuition as to why performing balancing steps over \emph{all} the random edges would probably not improve the running time.

Let us revert back to the greedy (local) proposals approach, used in Section~\ref{sec: det-alg}. 
Consider again the potential function
\[
\phi(V) = \sum_{\set{\set{u,v}\mid u,v\in V}} \size{w(u)-w(v)} ,
\]
which essentially sums up all load differences between all pairs of nodes. Balancing two nodes $u,v$ would reduce the potential by at least $\size{w(u)-w(v)}$ (and in some cases, exactly by this amount).
If one selects $k$ random terms in the sum above, the sum of these selected terms would be, in expectation, a $\Theta(k/n^2)$-fraction of the entire sum, due to the linearity of expectation.
So, balancing on the $k$ random edges added to the graph would roughly correspond to choosing $k$ random terms in the above sum,
guaranteeing an expected reduction of $\Theta(k/n^2)$-fraction of the current potential.
Such multiplicative reduction at each round would result in constant $\phi$ (and thus $\tau$-convergence for a constant $\tau$) within $\tilde{\Theta}\left(n^2/k\right)$ rounds.

This already shows that the expected complexity when using all the random edges at each iteration is the same as we achieve in Theorem~\ref{thm:smoothing_main}.
Moreover, it is not at all clear that such an intuition can be realized into an algorithm, for several reasons.
First, the nodes cannot distinguish the random edges from the edges added by the adversary, and there might be plenty more edges from the latter type.
Second, if $k =\omega(\sqrt{n})$, we expect some nodes to appear in multiple pairs (or edges), in which case the potential reduction for an edge $(u,v)$ is no longer  $\size{w(u)-w(v)}$.
Finally, the above intuition ignores integrality, which again prevents the full $\size{w(u)-w(v)}$ reduction.

\section{Discussion}

For $\tau=T/n$, Dinitz et al.~\cite{DinitzFGN17} prove a lower bound of $\Omega(n^2)$ rounds for matching-based algorithms in the dynamic setting. 
Our algorithm from Section~\ref{sec: det-alg} achieves this, though when each node may exchange load with two neighbors in every round and not with one.
The following question remains open: is it possible to design faster algorithms by utilizing two exchanges per round, or can the lower bound be extended to this case as well?

In Section~\ref{sec:impossibility} it is shown that no matching-based algorithm can solve \emph{integral} load balancing within bounded time. 
Both this result and the aforementioned lower bound of~\cite{DinitzFGN17} imply that in the dynamic setting, an adaptive adversary has a tremendous amount of power to prevent fast balancing.

However, we are able to show (Section~\ref{sec:smoothed analysis}) that these are, to some extent, merely pathological examples---and the first one being much more so.
To this end, we use \emph{smoothed analysis}, and show that for any smoothing parameter $k > 0$, balancing integral loads becomes possible;
for $k = \polylog(n)$, the convergence time is already $o(n^2)$. 
This applies also for the continuous case, as opposed to the  $\Omega(n^2)$ lower bound of the worst-case analysis.
Not only does our new algorithm properly leverage the noise in the network to deal with integrality, it also does it with grace: it is both deterministic and matching-based. 
Even in the continuous case, there is no algorithm that possesses both qualities and has a bounded convergence time.


\appendix
\section{Reducing a gap without knowing it}
\label{sec:balancing alg variant without knowing the gap}
We show an alternative realization of \alggapreduce{} that eliminates the need to globally defuse information through the network at the beginning (specifically the minimum and maximum loads). We show that one can simply guess $\psi$ and circumvent the notion of light and heavy nodes, incurring a slowdown of factor $\ln n$.

Indeed, the new algorithm is parametrized with $\psi$ and does the following: every node $u$ proposes to its heaviest neighbor $v$, as long as $w(v)\geq w(u)+\psi/2$.
As before, a node that receives at least one proposal, accepts the one from the lightest neighbor.
The main difference is that now balancing can be made between light and balanced nodes, as well as heavy and balanced nodes. 
Nevertheless, in each balancing step at least one heavy or light node is involved, and more importantly: the outcome is two balanced nodes. 
Thus, the sets $L_r$ and $H_r$ do shrink over time, but not necessarily both with every balancing step. 
We define a potential function for round $r$: $f(r) = \size{L_r} + \size{H_r}$, and we show that $f(r)$ drops between different phases. 
We are done once one of the sets is empty (regardless of the remaining potential).

The set $S$ of good potential edges (with gap at least $\psi/2$) in each phase still holds some guarantees, as long as we are not yet finished:
\[
\size{S}
\geq \size{L_r} \cdot \size{H_r} 
\geq \size{L_r} + \size{H_r} - 1.
\]
To see the central inequality, note that for any two numbers $a,b \geq 1$, $(a-1)(b-1)\geq 0$ and so $a\cdot b \geq a + b -1$.

Assume that some phase $j$ starts at round $r$ with $f(r) = j$; we analyze how many rounds are needed in order to reduce the potential to $j-1$ w.h.p.
Whenever the potential is $j$, we now have the weaker guarantee $\size{S} \geq j$ (compared to the guarantee $\size{S} \geq j^2$ we had in the main version).
Over all phases, noting that the potential is always $j \leq n$, we now have
\begin{align*}
    \sum_{i=1}^{n} t_i 
    &=\sum_{j=1}^{n} \frac{2n^2 \ln (n \ln T)}{c_1 k \cdot j} \\
    &=\frac{2n^2 \ln (n \ln T)}{c_1 k} \sum_{j=1}^{n} \frac{1}{j} \leq \frac{2n^2 \ln (n \ln T)}{c_1 k}  \ln n,
\end{align*}
with the last inequality relies on the harmonic number.

\medskip
As before, load balancing is done via multiple executions of \alggapreduce{} (using the new variant instead).
Unlike before, each such execution is now parametrized with $\psi$, to dictate nodes to balance with neighbors only when the load gap is at least $\psi/2$.
We therefore run the variant with decreasing values of $\psi$, starting from $\psi = T$, then $\psi = 3T/4$, and all the way down to $\tau < \psi < 4/3\tau$, after which the network is $\tau$-converged.
As before, each iteration is guaranteed to sufficiently decrease the gap with probability at least $1-1/(n\log^2 T)$, and after $O(\log T)$ such iterations, the network is $\tau$-converged with probability $1-1/n$.
\end{document}